\documentclass[12pt,a4paper]{article}
\pdfoutput=1
\usepackage{graphicx}
\usepackage{amsmath}
\usepackage{amsthm}
\usepackage{thmtools}
\usepackage{url}
\usepackage{alg,algcom}
\usepackage{color}
\usepackage[margin=1in]{geometry}
\usepackage{xspace}
\usepackage{enumitem}
\usepackage{amssymb}
\usepackage{euscript}
\let\script\EuScript
\def\todayst{\ifcase\day\or
 1st\or 2nd\or 3rd\or 4th\or 5th\or
 6th\or 7th\or 8th\or 9th\or 10th\or
 11th\or 12th\or 13th\or 14th\or 15th\or
 16th\or 17th\or 18th\or 19th\or 20th\or
 21st\or 22nd\or 23rd\or 24th\or 25th\or
 26th\or 27th\or 28th\or 29th\or 30th\or
 31st\fi
 ~\ifcase\month\or
 January\or February\or March\or April\or May\or June\or
 July\or August\or September\or October\or November\or December\fi
 \space \number\year}
\def\today{\ifcase\day\or
 1\or 2\or 3\or 4\or 5\or
 6\or 7\or 8\or 9\or 10\or
 11\or 12\or 13\or 14\or 15\or
 16\or 17\or 18\or 19\or 20\or
 21\or 22\or 23\or 24\or 25\or
 26\or 27\or 28\or 29\or 30\or
 31\fi
 ~\ifcase\month\or
 January\or February\or March\or April\or May\or June\or
 July\or August\or September\or October\or November\or December\fi
 \space \number\year}

\makeatletter
\def\makeop#1{\@namedef{#1}{\mathop{\operator@font #1}\nolimits}}
\makeatother
\makeop{SL}
\makeop{GL}
\makeop{PSL}
\makeop{tr}
\makeop{Stab}
\makeop{id}

\def\t/{Teich\-m\"uller}

\def\reals{\mathbb{R}}

\let\umlaut\"

\def\umlaut{\"}
\def\set#1{\left\{#1\right\}}
\def\bracket#1{\left(#1\right)}
\newcommand{\setN}{\ensuremath{\set{1,\cdots,N}}\xspace}
\newcommand{\bv}{\ensuremath{\mathrm{bv}}}

\newcommand{\sR}{\ensuremath{\script R}\xspace}
\newcommand{\sI}{\ensuremath{{\script I}_r}\xspace}

\newcommand{\sU}{\ensuremath{\script U_r}\xspace}
\newcommand{\sW}{\ensuremath{\script W_r}\xspace}
\newcommand{\sT}{\ensuremath{\script T_r}\xspace}
\newcommand{\mU}{{\mathrm U}}
\newcommand{\mL}{{\mathrm L}}
\newcommand{\mR}{{\mathrm R}}
\newcommand{\mW}{{\mathrm W}}
\newcommand{\mj}{{\mathrm j}}
\newcommand{\dr}{\mbox{$d$-rectangle}\xspace}
\newcommand{\drs}{\mbox{$d$-rectangles}\xspace}
\newcommand{\dimx}{\ensuremath{\mathtt{dim}}\xspace}
\newcommand{\vol}{\ensuremath{\mathtt{vol}}\xspace}
\newcommand{\FWU}{F_{\mathit{WU}}}
\newcommand{\tB}{\ensuremath{\mathtt{BSZ}}\xspace}
\newcommand{\tD}{\ensuremath{\mathtt{DLL}}\xspace}
\newcommand{\tI}{\ensuremath{\mathtt{ILLEGAL}}\xspace}

\newcommand{\tS}{\ensuremath{\mathtt{S}}\xspace}

\newcommand{\tU}{\ensuremath{\mathtt{UNDEF}}\xspace}

\newcommand{\tn}{\ensuremath{\mathtt{next}}\xspace}
\newcommand{\tp}{\ensuremath{\mathtt{prev}}\xspace}

\newcommand{\algendif}{\algend\textbf{endif}\\}
\newcommand{\algendfor}{\algend\textbf{endfor}\\}
	\newcommand{\algcomm}[1]{ \algcomment{/* }{\hspace{3.5mm}* }{#1}{ */}}
\newcommand{\algelsecomm}[1]{\algend\textbf{else}\algcomm{#1}\\\algbegin}
\theoremstyle{plain}
\newtheorem{theorem}{Theorem}[section]
\newtheorem{lemma}[theorem]{Lemma}
\newtheorem{proposition}[theorem]{Proposition}
\newtheorem{corollary}[theorem]{Corollary}

\theoremstyle{definition}
\newtheorem{definition}[theorem]{Definition}
\newtheorem{convention}[theorem]{Convention}

\newtheorem{rules}[theorem]{Rules}
\newenvironment{defn}%
{\ignorespacesafterend\begin{definition}}%
{\nolinebreak\qedsymbol\end{definition}}%
\title{Maximizing the area of intersection of rectangles}
\author{D.B.A. Epstein\footnote{Mathematics Institute, University
of Warwick.
Partial support from BBSRC Grant BB/K018868/1.}
\and M.S. Paterson\footnote{Computer Science Department, University
of Warwick.
Partial support from Centre for Discrete Mathematics and its Applications
(DIMAP).}}
\date{24 March 2017}
\begin{document}
\maketitle

\begin{abstract}
This paper attacks the following problem. We are given a large number $N$
of rectangles in the plane, each with horizontal and vertical sides,
and also a number $r<N$. The given list of $N$ rectangles may contain
duplicates. The problem is to find $r$ of these rectangles, such that,
if they are discarded, then the intersection of the remaining $(N-r)$ rectangles
has an intersection with as large an area as possible.
We will find an upper bound, depending only on $N$ and $r$, and not
on the particular data presented, for the number of steps
needed to run the algorithm on (a mathematical model of) a computer.
In fact our algorithm is able to determine, for each $s\le r$, $s$
rectangles from the given list of $N$ rectangles, such that the
remaining $(N-s)$ rectangles have as large an area as possible, and this
takes hardly any more time than taking care only of the case $s=r$.
Our algorithm extends to \drs---analogues of
rectangles, but in dimension $d$ instead of in dimension 2. Our method is
to exhaustively examine all possible intersections---this is much faster than
it sounds, because we do not need to examine all $\binom Ns$ subsets
in order to find all possible intersection rectangles.
For an extreme example, suppose the rectangles are nested,
e.g., concentric squares of distinct sizes, then the only intersections examined
are the smallest $s+1$ rectangles.
\end{abstract}

\section{Background}
This problem arose from the use of a novel robotically controlled
microscope technology
(\textit{Toponome Imaging System}, abbreviated to \textit{TIS}),
invented by Walter Schubert (see
\cite{Schubert.etal:2006:Analyzing}).
TIS records, pixel by pixel, the location and abundance of
proteins in a tissue section, thus co-locating many different
proteins in the same pixel.
Since proximity is necessary (though not sufficient) for interaction,
this gives a powerful new method for discovering protein complexes and
protein networks, particularly since the anatomy of the section is not
disrupted by the TIS process.
A TIS run results in a large number (typically, hundreds) of different
images of the same tissue section.

Such sequences of images may issue from the TIS machine approximately aligned,
but they will seldom be completely correctly aligned without further processing,
because physical reality, for example changes in temperature or
vibrations from a passing heavy goods vehicle, ensures
that camera measurements are never perfect.
In order to obtain good information about co-location of proteins and
other biomolecules, alignments that are as accurate as possible
must first be achieved. In \cite{Raza.etal:2012:RAMTaB:}, it is
explained how this can be done.
The only adjustments necessary to achieve good alignments turn out to be
rigid translation of the different images against one fixed image.

\begin{figure}[tbp]
\begin{center}
\includegraphics[height=40mm]{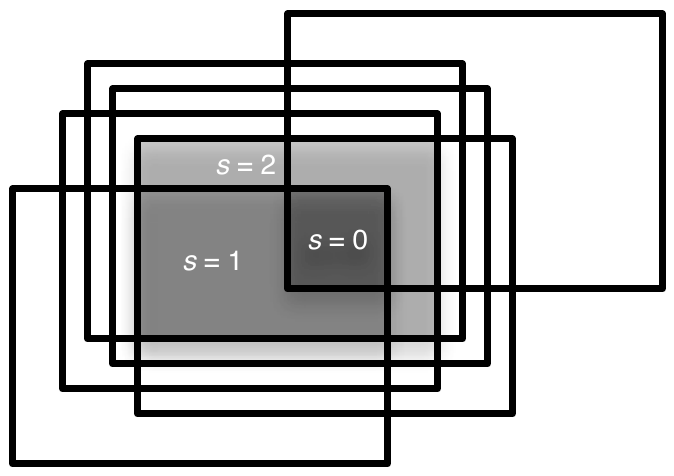}
\caption{$N=6,\: r=2$. Shading shows maximal intersections for $s=0,1,2$.}
\label{fig1}
\end{center}
\end{figure}
After alignment, some pixels will appear in all $N$ images, and others
may appear in only $(N-1)$ images, or in $(N-2)$ images, and so on, as shown
in Figure~\ref{fig1}.
If we insist that a pixel can
be analyzed only if it records all these signals simultaneously, then
we have to restrict to the region where all $N$ images overlap.
However, if one is prepared to discard completely one or more
images, then the set of pixels which figure in all remaining images
will be larger, much larger if one is discarding very seriously
misaligned images.
In practice one would normally decide to jettison poorest
quality images first, so the algorithm presented here, though useful,
is not the final word.

For a similar scenario in more conventional photography, suppose
one wants an image of St.~Mark's Basilica in Venice, but without
the tourists and pigeons. To achieve this, one might take snapshots
periodically from the same spot over an extended period. The
tourists and pigeons can then be eliminated by aligning the images
and then averaging the results over
each pixel in the intersection of the images. Using our algorithm
one can quickly determine which images to discard, in order to
obtain a final image of desired area.

\section{Statement of result}
We start by fixing our notation and clarifying the meaning of the terms
we use. In general, we will work in dimension $d$. In the discussion
above, it was assumed that $d=2$, but our algorithm will work
well in general.

\begin{defn}\label{interval defn}
Given an ordered pair $(\lambda,\rho)$ of reals numbers, we can form the
space $\set{x:\lambda\le x\le \rho}\subseteq \reals$.
If $\lambda < \rho$, this is a closed interval; if $\lambda=\rho$,
it is a point;
and if $\lambda>\rho$ it is the empty set. For the purposes of this paper, it
is convenient to have a concept that includes
the {\it endpoints} $\lambda$ and $\rho$ as well as the underlying
space.
Refraining from pedantic rigour, we will use the word {\it interval}
as though it is nothing
more the underlying space, but
nevertheless, we will feel free to extract the ``endpoints'' whenever
necessary, even when $\lambda>\rho$.
We will denote such an ``interval'' by $[\lambda,\rho]$, even if
$\lambda > \rho$.
\end{defn}

\begin{defn}\label{rectangle defn}
By a {\it \dr}, we mean a product
$$[\lambda_1,\rho_1]\times\dots\times [\lambda_d,\rho_d]\subseteq \reals^d.$$
With the same level of
informality as in Definition~\ref{interval defn}, we are able to
extract the $2d$ endpoints, even when the underlying space is the
empty set. Given a \dr $R$ as above, suppose that $\lambda_i <
\rho_i$ for exactly $k$ values of $i$, and that $\lambda_i=\rho_i$
for the other values of $i$. Then we call the $R$ a {\it $k$-dimensional
rectangle} and define $\dimx(R)=k$.
\end{defn}

Note that a \dr may or may not be a $d$-dimensional rectangle.
In general the underlying space
of a \dr may be the empty set or a point or a $k$-dimensional
rectangle in $\reals^d$.
A \dr can be presented as a $2d$-dimensional row vector
\begin{equation}\label{matrix rectangle eqn}
R =
\begin{bmatrix}
\lambda_1 & \rho_1 \dots & \lambda_d & \rho_d
\end{bmatrix}.
\end{equation}

\begin{defn}\label{volume defn} Given a \dr $R$ as above, we
extend the definition of its volume slightly. The {\it $d$-volume} of
$R$ is a pair $(\dimx(R),\vol(R))$ where:
\begin{enumerate}
\item
If, for each $i$, $\lambda_i < \rho_i$), then
$$
\dimx(R) = d \text{ and } \vol(R)
=(\rho_1-\lambda_1)\times \dots \times (\rho_d-\lambda_d) >0 .
$$
\item If $d>k> 0$ and $R$ is a $k$-dimensional rectangle in $\reals^d$
(as in Definition~\ref{rectangle defn}),
then $\dimx(R)=k$ and $\vol(R)$ is equal to the volume of
$R$, regarded as a $k$-rectangle in $\reals^k$.
For example, if $k=1$, then $R$ is an interval in $\reals^d$ and
$\vol(R)$ is its usual length.
\item If $R$ is a 0-rectangle (that is, a point), then
$(\dimx(R),\vol(R)) = (0,0)$.
\item If, for some $i$, $\lambda_i>\rho_i$, so that the underlying
point set of $R$ is $\emptyset$, then
$$(\dimx(R),\vol(R))=(-1,0).$$
\end{enumerate}
If $d=1$, this gives the length
of the interval, and if $d=2$ the area of the rectangle.
If $d$-volumes need to be ordered, then we use lexicographical
ordering on the pair $(\dimx(R),\vol(R))$.
\end{defn}

\begin{convention}\label{data}
Throughout this paper (unless otherwise stated) $d>0$
will denote the dimension of the euclidean space $\reals^d$ in which
we are working, and \sR will denote a specific ordered list of $N$
\drs. This list can equally well be represented by an $N\times 2d$
real matrix.
We fix an integer $r$ with $0< r<N$.
\end{convention}

Our result can now be stated. Its proof will occupy most of the rest of
this paper.
\begin{theorem}\label{main}
We are given data as in Convention~\ref{data}.
Then there
is an algorithm taking at most
$$
O\bracket{d\binom{r+2d}{2d} + dN}
$$
steps, that determines, for each $s$ with
$0<s\le r$, exactly $s$ entries to delete from the list \sR, so that
the list of the remaining $(N-s)$ \drs,
represented by the remaining $(N-s)$ elements of \sR,
has an intersection whose $d$-volume
is maximal, as we vary over the $\binom{N}{s}$ possible choices of
$s$ elements of \sR.
The constant in our bound for the number of steps is independent of
$N$, $r$ and $d$.
\end{theorem}
\begin{corollary}\label{maincor}
For any fixed dimension $d$, the running time of the algorithm is $O(r^{2d} + N)$.
\end{corollary}

In applications such as that of microscopy, which was our primary motivation,
many rectangles may coincide in some coordinates. Our algorithm is designed
to deal efficiently with this situation. Several of the complications result from this.

We have been thinking mainly of the case when $r$ is small compared
with $N$.
For $r$ near $N$, it is possible to achieve better results with a brute
force search. For example, if $r=N-1$, then one deletes all except one
\dr. So one need only compute the volume of each of the given
\drs, and one then takes the maximum of these. This takes
$O(dN)$ steps. We will not bother to make explicit other similar brute force
results when $r$ is nearly as big as $N$.

To be explicit,
we need a model of computation: we will assume the
usual model of a RAM machine. This means that the number of memory
cells is as large as we need for the computation, and that each of
these is accessed in unit time. Moreover, memory cells are assumed to
be large enough so that any particular integer that needs to be stored
can be stored in a single cell. Formally, we do not work with floating
point numbers. If necessary, these can be converted to integers by a
suitable change of scale. However, it is convenient and less confusing
to pretend that we can work with arbitrary real numbers.
Arithmetical operations including simple comparisons are assumed to
take unit time. This gives a reasonably realistic mathematical
model of a computer for most practical purposes.

\section{Some notations and definitions}
We can think of
a \dr as an element in $\reals^{2d}$, a
row vector as in Equation~\ref{matrix rectangle eqn}.
$\sR$ is then a map $\sR:\setN\to\reals^{2d}$.
In particular, repetition is allowed.
\begin{convention}\label{lists}
Formally speaking, a list is a map, like $\sR$. However, it is
convenient to use the same notation also for the image of
the map. The advantage of this abuse of
notation is that it allows us to write $R\in\sR$ for any element $R$ in
the list. This paper will contain several lists, and we allow the same
abuse of notation in each case. We will disambiguate if we
think there is a possibility of confusion.
\end{convention}

For $1\le j\le N$, we write $R_j$ instead of $\sR(j)$,
and,
following Equation~\ref{matrix rectangle eqn},
we will consistently use the notation:.
\begin{equation}\label{matrix rectangles eqn}
R_j =
\begin{bmatrix}
\lambda_{j,1} & \rho_{j,1} \dots & \lambda_{j,d} & \rho_{j,d}
\end{bmatrix}
=
\begin{bmatrix}
r_{j,1} & \cdots & r_{j,2d}
\end{bmatrix}
\end{equation}

\begin{defn}\label{rectangle indices defn}
We often refer to elements of \setN as {\it rectangle indices}, to
remind ourselves that they are indices for our list $\sR$ of
\drs.
\end{defn}

\begin{defn}\label{intersection defn}
Let $L\subseteq \setN$ be non-empty.
The {\it intersection} of the family $\bracket{R_j}_{j\in L}$
is defined as follows.
For each $i$ with $1\le i\le d$, let
$M_i=\max_{j\in L}\lambda_{ji}$ and let
$m_i=\min_{j\in L}\rho_{ji}$.
Then the {\it intersection} is the \dr
\begin{equation}\label{intersection eqn}
{\mR}(L) = \begin{bmatrix}M_1 & m_1 & \dots & M_d & m_d\end{bmatrix}.
\end{equation}
Equation~\ref{intersection eqn} is the result of a trivial calculation
when ${\mR}(L)$ is not empty, and is a definition when
the point set of ${\mR}(L)$ is empty.
${\mR}(L)$ is called an {\it intersection rectangle}.
We denote by \sI the set of non-empty intersection rectangles
$({\mR}(L)\neq\emptyset)$, where $L\subseteq \setN$ and $|L|\ge N-r$.
\end{defn}

We will find an efficient way to produce, one by one, the
elements of \sI. We will then be able to calculate, for each $R\in\sI$,
the $d$-volume of $R$ and compare these results to find the best one.
The number of subsets like $L$ above, giving rise to elements of \sI
in Definition~\ref{intersection defn},
could be as large as $\sum_{s=0}^r \binom Ns$, potentially a large number.
To reduce this number, note that a given element $R\in \sI$ can in general
satisfy $R= {\mR}(L)$ for many different $L\subset \setN$.
For any such $L$, we define $\mL(R)$, by
\begin{equation}\label{L function}
\mL(R)= \set{j|R\subseteq R_j},
\end{equation}
which is the largest subset of the given set of
rectangles with intersection $R$.
Then $L\subseteq \mL(R) = \mL(\mR(L))$, and $|\mL(\mR(L))| \ge |L| \ge N-r$.

We are thinking of $r$ as small compared with $N$, and so it is easier
to work with the complement of $\mL(R)$ in \setN. We define
\begin{equation}\label{U function}
\mU(R)=\set{j|R\nsubseteq R_j}
\end{equation}
so that $|\mU(R)| \le r$.
\begin{defn}\label{U sub r defn}
We define
$$\sU = \set{\mU(R): R\in\sI}.$$
The map $\mU:\sI \to \sU$ is clearly a bijection.
\end{defn}

An essential feature of our solution will be the use of different
orderings on the set \setN, one for each $i$ with $1\le i \le 2d$.
To understand better why and how orderings are important,
we discuss the special case $d=1$. In this case, we are given a list of $N$
intervals $\sR=\bracket{I_1,\ldots,I_N}$, with
$I_j=[\lambda_j,\rho_j]$.
To avoid irrelevant detail, we assume for
the moment that the $2N$ endpoints of these intervals are all
distinct, and that the intersection of the $N$ intervals is not empty.
We define $u,v \in \setN$ so that $\lambda_u=\max_j\lambda_j=m$ and
$\rho_v=\min_j\rho_j=M$, as in Definition~\ref{intersection defn}.
Unless $I_u$ is removed from \sR, the lefthand
endpoint of the intersection will continue to be $m$. Unless $I_v$ is
removed from \sR, the righthand endpoint of the intersection will
continue to be $M$.
So $I_u$ or $I_v$ must be removed (and remember that $u=v$ is a
possibility) if the length
of the intersection is to strictly increase. If $r\ge 2$, and we are
omitting $r$ intervals from \sR, then similar reasoning will apply to
subsequent removal of intervals. We have to look repeatedly for the
largest lefthand endpoint and, separately, for the smallest righthand
endpoint, in the remaining intervals.

Given a list \sR of $N$ \drs,
we now define, for each $i$ with $1\le i \le 2d$,
orderings $\prec_i$ and $\preccurlyeq_i$ on \setN, modelling the definitions on the
discussion in the previous paragraph.
\begin{defn}\label{orderings defn}
Let $j,k\in \setN$ and let $1\le t \le d$.
If $i=2t-1$, we define $j\prec_i k$ if $\lambda_{j,t} > \lambda_{k,t}$, and
$j\preccurlyeq_i k$ if $\lambda_{j,t}>\lambda_{k,t}$ or
if $(\lambda_{j,t}=\lambda_{k,t}\text{ and } j\le k)$.
If $i=2t$, we define $j\prec_i k$ if $\rho_{j,t} < \rho_{k,t}$, and
$j\preccurlyeq_i k$ if $\rho_{j,t}<\rho_{k,t}$ or if
$(\rho_{j,t}=\rho_{k,t}\text{ and } j\le k)$.
Then $\prec_i$ is a partial order and $\preccurlyeq_i$ is a total order on \setN.
For odd $i$, the orderings are (weakly) decreasing in the $\lambda$'s, and,
for even $i$, they are (weakly) increasing in the $\rho$'s.
\end{defn}

\section{Properties of $X\in\sU$.}
We want to produce the elements $X\in\sU$ efficiently one by one. With
this in mind, let $X\subset\setN$ with $|X|\le r <N$, in which case we
find necessary and sufficient conditions for $X$ to be an element of \sU.

\begin{defn}\label{endpoint list defn}
An {\it endpoint list} is a non-empty subset of \setN, arranged in order,
smallest first, with respect to one of the total orderings $\preccurlyeq_i$ described in
Definition~\ref{orderings defn}.
We set $E_i$ equal to the set \setN, arranged in order according
to $\preccurlyeq_i$, so that $E_i$ is an endpoint list.
\end{defn}

\begin{defn}\label{bv defn}
Let $X\subset\setN$, with $|X|\le r < N$.
For each $i$, let
$\mj(i,X)$ be the smallest element of $E_i$ that is not in $X$ (which must
exist since $X$ has fewer than $N$ elements).
We define $\bv(i,X)=r_{\mj(i,X),i}$
(where $\bv$ stands for {\it barrier value}).
We set
\begin{equation}\label{v(X) eqn}
\bv(X) =(\bv(1,X),\cdots ,\bv(2d,X)).
\end{equation}
\end{defn}

\begin{lemma}\label{U description}
Suppose $X\subset\setN$, with $|X| \le r < N$.
$X\in\sU$ if and only if the following two conditions are satisfied:
\begin{enumerate}[label=\thetheorem.\arabic{*}:\ ,ref=\thetheorem.\arabic{*},
leftmargin=0cm,itemindent=.5cm,labelwidth=\itemindent,labelsep=0cm,align=left]
\item\label{U condn: bv inequality} For each $i$ with $1\le i\le d$,
$\bv(2i-1,X)\le \bv(2i,X)$.
\item\label{U condn: prec}
For each $j\in X$, there is an $i$ such that $j\prec_i \mj(i,X)$.
\end{enumerate}
Under these conditions, $X=\mU(\bv(X))$.
\end{lemma}
\begin{proof}
Suppose first that $X\in\sU$, or equivalently that
$X=\mU(R)$ for some $R\in\sI$, where
$$\emptyset\neq R=\begin{bmatrix}r_1&r_2&\dots&r_{2d-1}&r_{2d}\end{bmatrix}.$$
For any $j\notin \mU(R)$, we have $R\subseteq R_j$, and so,
for $1\le i\le d$,
\begin{equation*}
r_{j,2i-1}\le r_{2i-1} \le r_{2i} \le r_{j,2i}.
\end{equation*}
This holds in particular with $j=\mj(2i-1,X)$ and with $j=\mj(2i,X)$,
giving rise to
\begin{equation}\label{bv inequalities}
\bv(2i-1,X)\le r_{2i-1} \le r_{2i} \le \bv(2i,X),
\end{equation}
proving the first condition in the statement of Lemma~\ref{U
description}.

Continuing with the hypotheses $R\in\sI$ and $X=\mU(R)$,
we next show that $\bv(X)=R$, when we regard both as elements of $\reals^{2d}$.
Suppose $1\le i\le 2d$. Since $R=\bigcap\set{R_j:R\subseteq R_j}$,
we know that $r_i = r_{j,i}$, for some $j$ such that $R\subseteq
R_j$, that is, for some $j\notin X$.
By Equation~\ref{bv inequalities}, if $i$ is odd, we have
$$r_i=r_{j,i} \ge \bv(i,X) = r_{\mj(i,X)}.$$
But $\mj(i,X)$ is, by definition, the smallest element not in $X$,
and the ordering is reversed for odd $i$, according to
the ordering $\preccurlyeq_{i}$ (see Definition~\ref{orderings defn}).
So a strict inequality
$\bv(i,X)<r_{j,i}$ is impossible, and we must have the equality
$\bv(i,X)= r_{i}$. Similarly if $i$ is even. This proves that
$R=\bv(X)$.
Therefore $X=\mU(R)=\mU(\bv(X))$.

We continue to assume that $X=\mU(R)$, and let $j\in X$.
Then $R\nsubseteq R_j$.
It follows that, for some $i$, $[r_{2i-1},r_{2i}]\nsubseteq
[r_{j,2i-1},r_{j,2i}]$.
This implies that
$r_{2i-1}<r_{j,2i-1}$ or $r_{2i}>r_{j,2i}$.
We have shown that $R=\bv(X)$, so this means that $j\prec_{2i-1}
\mj(2i-1,X)$ or $j\prec_{2i} \mj(2i,X)$, proving the second condition in
the statement of Lemma~\ref{U description}.

Now let us assume the two conditions in the statement of the lemma,
and prove that $X\in\sU$.
Let $R=\bv(X)$ be a \dr.
By Condition~\ref{U condn: bv inequality}, $R\neq \emptyset$.
If $j\notin X$, then, for each $i$, $\mj(i,X) \preccurlyeq_i j$.
Therefore $\bv(2i-1,X) \ge r_{j,2i-1}$ and $\bv(2i,X)\le r_{j,2i}$.
It follows that $R\subseteq R_j$.
In particular, for each $j\notin X$ and each $i$ with $1\le i\le 2d$,
$R\subseteq R_{\mj(i,X)}$, so that $R\subseteq \bigcap_i R_{\mj(i,X)}$.
On the other hand, Definition~\ref{intersection defn} shows that
$\bigcap_i R_{\mj(i,X)}\subseteq R$.
We deduce that $R=\bigcap_{j\notin X} R_j$, which shows that
$R\in\sI$.

If, on the other hand, $j\in X$, then, by hypothesis, for some $i$, $j\prec_i
\mj(i,X)$.
Then $r_{j,i} > \bv(i,X)$ if $i$ is odd, and $r_{j,i} < \bv(i,X)$ if $i$
is even.
It follows that $R\nsubseteq R_j$. Therefore
$X=\set{j:R\nsubseteq R_j}$ and $X\in \sU$.
\end{proof}

\section{Program paths}\label{Program paths}
Before describing the algorithm that produces all the elements of
$\sU$ one by one, we describe how to generate one particular
$U\in\sU$.
From Lemma~\ref{U description}, $U=\mU(\bv(U))$.

For $1\le i\le 2d$,
we write $U_i=\set{j|j\prec_i \mj(i,U)}$, so that, by Lemma~\ref{U
description}, $U = \bigcup_i U_i$.
For $1\le i\le 2d$ let
\begin{equation}\label{U sup i}
U^i=U_i\setminus\bigcup_{1\le k<i}U_k.
\end{equation}
$U$ is equal to the disjoint union of the $U^i$.
We generate first the elements of $U^1$, then the elements of $U^2$,
and so on.

We prefer to think of these rectangle indices in $U$ as being {\it
discarded} or deleted, rather than generated, since we
are thinking of generating the complement of $U$,
$L=\setN\setminus U$, and then taking the intersection of the \drs in $L$.
We start with \setN, discard the rectangle indices in $U^1$, then the
rectangle indices in $U^2$, and so on.
To completely specify the order in which the rectangle indices in $U$ are
discarded, we insist that the elements of $U^i$ are discarded in
increasing $\preccurlyeq_i$ order, that is, least first.

\begin{defn}\label{focus defn}
By a $\Delta$-{\it move} we mean one of the deletions just mentioned.
While working with $U^i$, $\Delta$
deletes the $\preccurlyeq_i$-smallest element of $E_i$ remaining.
Until all elements of $U^i$ are deleted, this smallest element is an
element of $U^i$.
While carrying out such deletions, we say that the algorithm is
{\it focussed at $i$}.
An $S$-{\it move} moves from the {\it focus index} $i$ to $(i+1)$.
In our process, there are only these two types of move, and the
process starts with focus at $i=1$.
$\Delta$ and $S$ are mnemonic for ``delete'' and ``shift index by one''.
At the end of the process, a final $S$-move causes the process to stop.
\end{defn}

During the algorithm the $E_i$ keep on changing.
In order to
keep track of what is going on, we introduce a time variable $t$ which
increases by one at each move.
Let $\delta(t,U)$ be the number of $\Delta$-moves and
$s(t,U)$ the number of $S$-moves
during the first $t$ units of time, so that $t=\delta(t,U)+s(t,U)$.
At time $t$, the focus is at $i=s(t,U)+1$.
Recall from Definition~\ref{endpoint list defn} that $E_i$
consists of the elements of \setN listed in ascending order according
to $\preccurlyeq_i$.
We write $E_i(t,U)$ for $E_i$ at time $t$, when
all the lists $E_i(t,U)$ have exactly the same
set of $N-\delta(t,U)$ entries, though the orders in which these entries
appear are, in general, different.
When a $\Delta$-move deletes an element of $E_i(t,U)$ during focus at
$i$ (always the smallest element according to $\preccurlyeq_i$), then that element is
also deleted from all the $E_k(t,U)$, for $1\le k \le 2d$.

The process starts at $t=0$ with focus at $i=1$.
For each $k$ with $1\le k\le 2d$, $E_k(0,U)=\setN$.
For $1\le i\le 2d$, let $t_i$ be the first time at which the focus
moves to $i$ (so $t_1=0$),
and let $t_{2d+1}$ be the time when the process ends.
At time $t_i$, the set of rectangle indices so far deleted
from \setN and from the various $E_\ell$,
is equal to $\bigcup_{1\le k<i}U^k$.
From the definition of $U_i$ immediately preceding Equation~\ref{U sup
i}, we know that
\begin{equation}\label{Ui}
U^i = \set{j|j\in E_i(t_i,U)\text{ and }j\prec_i \mj(i,U)}
\end{equation}
(see Convention~\ref{lists} for the meaning of $j\in E_i$), so that
the $\preccurlyeq_i$-smallest element of $U^i$, which is the next rectangle
number to be discarded (assuming $U^i$ is not empty),
is also the $\preccurlyeq_i$-smallest element of $E_i(t,U)$.
Therefore, $|U^i|$ successive $\Delta$-moves will delete exactly the elements
of $U^i$, and, following this, an $S$-move transfers focus to $(i+1)$,
or, if $i=2d$, $S$ causes the process to stop.

\begin{rules}\label{Rules}
It is convenient to think of the process $\pi(U)$ just described, when
$U\in\sU$, as a word
$\mW(U)$ in the terms $S$ and $\Delta$, read from left to right.
Let $u(i)=\left| U^i\right|$ be the number of elements in $U^i$.
Note that $\sum_i u(i) \le r$.
We define
\begin{equation}\label{W(U) defn}
\mW(U)= \Delta^{u(1)}S\Delta^{u(2)}S\cdots S\Delta^{u(2d)}S.
\end{equation}

We will now present several rules on words $w$ in $S$ and $\Delta$ that
every $w=\mW(U)$ has to obey.
Once the rules have all been made explicit,
we will define $\sW$ to be the set of legal words, that is, the set of all words
satisfying all the rules.
We will then show that the map $\FWU:\sU \to\sW$ given by
$\FWU(U)=\mW(U)$ is a bijection, so that the rules will be not only
necessary but also sufficient for a word $w$ to be equal to $\mW(U)$ for
some $U\in\sU$.
This will reduce our task to producing an algorithm that produces
one by one the set of all legal words.

We already have three rules, satisfied by any word $w=\mW(U)$ for some
$U\in\sU$:
\begin{enumerate}[label=\bf Rule~\Alph{*}:\ ,
ref=Rule~\Alph{*},
leftmargin=0cm,itemindent=.5cm,labelwidth=\itemindent,labelsep=0cm,align=left]
\item\label{2d rule} $w$ contains exactly $2d$ terms equal to $S$.
\item\label{end rule} $w$ ends with an $S$.
\item\label{r rule} $w$ contains at most $r$ terms equal to $\Delta$.
\end{enumerate}
Rules A and B mean that each coordinate direction is considered and Rule C
ensures that at most $r$ deletions are used.

\begin{defn}\label{pi(w) and E(t,k,w)}
Now suppose a word $w$ in $\Delta$ and $S$ satisfies
\ref{2d rule}, \ref{end rule} and \ref{r rule}.
We obtain a process $\pi(w)$, analogous to the process $\pi(U)$ (see
first sentence of \ref{Rules}), that
reads $w$ from left to right, taking action as appropriate for each
term read.
More precisely, we proceed as follows.
If $t\le \mathrm{length}(w)$,
we factorize $w=w_t.v_t$, where the length of $w_t$ is $t$.
We denote by $\delta(t,w)$ the number of terms equal to $\Delta$ in
$w_t$, and by $s(t,w)$ the number of terms equal to $S$ in $w_t$.
After reading $w_t$, the focus is at $i=i(t,w)=s(t,w)+1$.
In particular, the focus is at $i=1$ when $t=0$.
We denote by $E_k(t,w)$ the list $E_k$ after the $\delta(t,w)$
deletions that occur while $w_t$ is being read.
With the focus at $i$, let $j$ be the $\preccurlyeq_i$-smallest element of $E_i(t,w)$.
If the first term of $v_t$ is $\Delta$, then $j$ is deleted from each
$E_k(t,w)$, giving the endpoint list $E_k(t+1,w)$.
If the first term of $v_t$ is $S$,
then focus is transferred from $i$ to $(i+1)$, and, for
each $k$ with $1\le k \le 2d$, $E_k(t+1,w)=E_k(t,w)$.
Let $U(w)$ be the set of all the rectangle indices $j$ deleted
while $w$ is read.
We are interested only when $U(w)\in\sU$, which is not in general the
case for an arbitrary word $w$ satisfying only the three rules above.
\end{defn}

We will now describe further rules that take account of coincident
coordinates, and will prove in
Lemma~\ref{rules true} that these rules are satisfied by $w=\mW(U)$,
whenever $U\in \sU$.
It will turn out that, conversely,
if $w$ satisfies all the rules we present, then there is a
unique $U\in \sU$ such that $w=\mW(U)$.

\begin{enumerate}[resume*]
\item\label{barrier constant rule}
As rectangles are discarded, barrier values may change.
A lefthand endpoint barrier value will not increase and may decrease.
A righthand endpoint barrier value will not decrease and may increase.
This rule will ensure that any rectangle whose removal would change
the barrier value at $i$ must be removed when the focus is no later than $i$.
Let $\mj(t,k,w)$ be the $\preccurlyeq_k$-smallest element of $E_k(t,w)$, and let
$\bv(t,k,w)=r_{\mj(t,k,w),k}$.
Suppose that $2\le i\le 2d$ and that
$t_i$ is the first time at which the focus is at $i$.
The condition on $w$ is that,
if $1\le k<i$, then $\bv(t,k,w)$ is independent of $t$ provided that
$t_i\le t\le \mathrm{length}(w)$.
We denote this constant value by $\bv(k,w)$.

\item\label{block rule}
Fix $i$ with $1\le i\le 2d$.
Let $A$ be an endpoint list with ordering $\preccurlyeq_i$.
By an $A$-{\it block} we mean a non-empty set of
entries in $A$ of the form
$$B=\set{j|j\in A \text{ and } r_{ji}=x}$$
for some $x\in\reals$.
Any $j\in A$ is contained in a unique $A$-block.
Such a block can be of length one, but blocks can also be of any
length up to $N$, as would happen, for example, if all of the
\drs in the list $\sR$ were equal to each other.
We denote by $B_k(t,w)$ the initial $E_k(t,w)$-block, and suppose $t$
is such that the focus is at $k$.
Let $b(t,w)=|B_k(t,w)|$. Then $b(t,w)>0$.
The condition on $w=w_t.v_t$ is that,
if $v_t$ starts with $\Delta$, then $\delta(t,w)+b(t,w)\le r$ and $v_t$
starts with $\Delta^{b(t,w)}$.
In other words, if $v_t$ starts with a $\Delta$, then all the
elements of the initial block of $E_k(t,w)$ must be deleted
before the focus increases.
\item\label{barrier value inequality rule}
For each $k$ with $1\le k\le d$, $\bv(2k-1,w)\le \bv(2k,w)$.
This rule ensures that only non-empty intersections are reached,
and this can be checked as the focus leaves any even integer.
\end{enumerate}
\end{rules}

\begin{lemma}\label{rules true}
Let $U\in\sU$ and $w=\mW(U)$.
Then $w$ satisfies all the rules above.
\end{lemma}
\begin{proof}
From Equation~\ref{W(U) defn} we immediately see that
\ref{2d rule}, \ref{end rule} and \ref{r rule} are true.

To prove \ref{barrier constant rule}, recall that $\mj(k,U)$ is the
$\preccurlyeq_k$ smallest element of $E_k(0,w)\setminus U$
(see Definition~\ref{pi(w) and E(t,k,w)} for the meaning of $E_k(t,w)$).
Since only elements of $U$ are discarded while $w$ is being
read, $\mj(k,U)$ is never discarded.
Therefore, for all $t$, $\mj(k,U)$ is the $\preccurlyeq_k$
smallest element of $E_k(t,w)\setminus U$.
Therefore, for all $t$,
\begin{equation}\label{j inequality}
\mj(t,k,w)\preccurlyeq_k \mj(k,U).
\end{equation}
We claim that, for $t\ge t_{k+1}$,
\begin{equation}\label{independence of t}
\bv(t,k,w)= r_{\mj(t,k,w),k}= r_{\mj(k,U),k} = \bv(k,U),
\end{equation}
where the first and third equalities are definitions.
For, if Equation~\ref{independence of t} were false, then
Equation~\ref{j inequality} would give $\mj(t,k,w)\prec_k \mj(k,U)$.
Recall from the definition just before Equation~\ref{U sup i},
that $U_k=\set{j: j\prec_k \mj(k,U)}$, so that
$\mj(t,k,w)\in U_k$.
However, for $t\ge t_{k+1}$, all elements of $U_k$ have been
discarded, in particular $\mj(t,k,w)$,
This contradiction proves our claim, which proves \ref{barrier constant rule}.

Now we prove \ref{block rule}.
Choose $k$ with $1\le k \le 2d$, and suppose the focus is at $k$ at
time $t$.
Let $B$ be the initial block of $E_k(t,w)$, and let $x\in \reals$
be the real number such that
$ B = \set{j\in E_k(t,w) | r_{jk} = x}$.
The initial $\Delta$ in $v_t$ deletes some $j\in B\cap U^k$, so that
$j\prec_k \mj(k,U)$.
If $k$ is even, this means that $x=r_{j,k} < r_{\mj(k,U),k}$, and, if $k$ is
odd, $x=r_{j,k} > r_{\mj(k,U),k}$. It follows that, for all $j\in B$,
$j\prec_k \mj(k,U)$. But this means that $B\subseteq U_k$.
Since the focus is on $k$, $B\cap U^j=\emptyset$ for $j<k$.
We deduce that $B\subseteq U^k$, and so all elements of $B$ are
discarded immediately.
This completes the proof of \ref{block rule}.

\ref{barrier value inequality rule} follows from
Equation~\ref{independence of t} and Lemma~\ref{U description}, and
this completes the proof of Lemma~\ref{rules true}.
\end{proof}

\begin{proposition}
The map $\FWU:\sU\to \sW$, defined by $\FWU(U)=\mW(U)$, is a bijection.
\end{proposition}
\begin{proof}
Lemma~\ref{rules true} shows that we do indeed have $\mW(U)\in\sW$, for
each $U\in\sU$.

Conversely, given $w\in\sW$, we obtain a subset $U\subset \setN$ of
discards by means of the process $\pi(w)$ (Definition~\ref{pi(w) and
E(t,k,w)}).
We need to prove the two conditions of Lemma~\ref{U description}.

Let $V_i=\set{j: j \prec_i \mj(i,U)}$.
By the definition of $\mj(i,U)$ (see Definition~\ref{bv defn}),
$V_i\subseteq U$, so, in due course, all elements of $V_i$ are
discarded during the process $\pi(w)$.
It follows that, for $i$ fixed and $t$ sufficiently large (of course
$t\le r+2d$ throughout),
$\mj(t,i,w)=\mj(i,U)$,
and therefore $\bv(i,w)=\bv(i,U)$.
\ref{barrier value inequality rule} implies \ref{U condn: bv inequality}.

Let $j\in U$.
Then, for some $i$, $j$ is discarded while the focus is at $i$.
While the focus is at $i$, only elements $k\in U$ with $k\prec_i\mj(i,U)$
are discarded.
Therefore $j \prec_i \mj(i,U)$.
We claim that $r_{j,i} \neq r_{\mj(i,U),i}$.
For otherwise $j$ and $\mj(i,U)$ would be in the same $E_i$-block.
But then \ref{block rule} would imply that $\mj(i,U)$ is also
discarded, implying that $\mj(i,U)\in U$.
This contradiction implies that $j\prec_i \mj(i,U)$, which is
\ref{U condn: prec}.
We have proved the two conditions of Lemma~\ref{U description},
and so $U\in\sU$.
\end{proof}

\section{A brief sketch of the algorithm}
\label{brief sketch}
We have reduced our problem to that of finding an
algorithm that produces the words of $\sW$ one-by-one.
Consider the finite set $\script W(r)$ of words $w$ satisfying
\ref{2d rule}, \ref{end rule} and \ref{r rule}.
In the standard fashion, we form these words into a rooted tree
$\script T(r)$, with directed edges, each
edge labelled with either $\Delta$ or $S$, so that there is exactly one
path in $\script T$ from the root to a leaf for each word in $\script W(r)$.
There is a subtree $\sT$ corresponding to the words of $\sW$.

Our algorithm traverses a certain subtree of the tree
$\script T(r)$, in depth-first order, choosing to follow a $\Delta$-edge rather
than an $S$-edge, where there is a choice.
For efficiency reasons, we would prefer to arrange for this
subtree to be exactly the subtree $\sT$ of interest.
However, staying inside $\sT$ would require a substantial
amount of additional computation---this is because
\ref{barrier value inequality rule} cannot be easily checked
term by term.
Instead we explore a subtree of $\script T(r)$
that is somewhat larger than $\sT$.
As soon as we discover that we are outside $\sT$, we backtrack.

Our algorithm uses a recursive function call to perform a
depth-first search of $\sT$.
We maintain a global data structure enabling us to check
that the rules are being followed during the search.
We denote the recursive call by $F(p,s)$.
Here $p$ is the maximum
number of times we permit further rectangle indices to be discarded from
\setN, that is, the maximum number of terms $\Delta$ that the
algorithm may still generate, and $p$ is mnemomic for {\it
permissible discard}. Also $s$ is the number of endpoint lists
still to be examined, that is, the number of terms $S$
still to be generated, and $s$ is mnemonic for {\it shift}.
The behaviour of $F(p,s)$ depends on the state of the global data structure,
for example deciding whether to call $F(p-1,s)$ or
$F(p,s-1)$ or both.
The depth-first aspect of the search comes about since $F(p,s)$ calls
$F(p-1,s)$ before $F(p,s-1)$ if the data structure tells us to call
both.
A principle followed by our algorithm is that $F(p,s)$ alters the most
important part of the data
structure before calling $F(p-1,s)$ and/or $F(p,s-1)$.
On its return, $F(p,s)$ changes this part of the data structure
back again before exiting.
Thus, (the most important part of) the global data structure, when
$F(p,s)$
starts, is the same as the structure when $F(p,s)$ finishes.

The depth-first search starts by calling $F(r,2d)$.

\section{Selection versus sorting}
\label{selection versus sorting}
We will later describe in detail
the data structures used by our algorithm to check,
as rapidly as possible, that the rules are being followed.
In this section we discuss a minor improvement that can be made on the
most obvious approach.
The first three rules are enforced by the notation $F(p,s)$ and so
have no implications for the data structure.
The procedure described above requires us to
order \setN with respect to the $2d$-different
orderings $\preccurlyeq_i$ described in Definition~\ref{orderings defn}.
This takes $O(dN\log N)$ steps.

However, we can do a little better than this.
For each $k$ with $1\le k\le 2d$, let $G_k\subseteq \setN$ be
the ordered set of the first $r+1$ elements of \setN with respect to
the ordering $\preccurlyeq_k$.
It turns out (see Lemma~\ref{G rules}) that our algorithm can be
confined to consideration only of the ordered sets $G_k$, and that it
will then run in the same number of steps (or slightly fewer).
To compute the $G_k$ we proceed as follows.
For each $k$, find the $(r+1)$-st element $p_k\in E_k$.
An algorithm doing this is called a {\it
selection algorithm}, and it requires $O(N)$ steps, as described in
the next paragraph.
As an unordered set $G_k = \set{j:j\preccurlyeq_k p_k}$ is found in $O(N)$ steps.
$G_k$ is then sorted in $O(r\log r)$ steps.
To obtain all the $G_k$ as ordered sets takes
\begin{equation}\label{selection}
O(d(N+r\log r))
\end{equation}
steps.

There is a randomized selection algorithm
\cite{Hoare:1961:Algorithm}, due to Hoare, with an expected number of
steps equal to $O(N)$, though the worst-case bound is $O(N^2)$.
\cite{Blum.etal:1973:Time} contains an algorithm where the number of
steps is bounded by $O(N)$ steps even in the worst case, though the
constant hidden by the $O$-notation is not small.
The latter algorithm only rarely outperforms the Hoare algorithm,
and the Hoare algorithm also has the advantage that it can be
carried out in place, that is, without needing more space than is
needed by the input data.
There is a good article about this topic, giving Hoare's code, at
\url{http://en.wikipedia.org/wiki/Selection_algorithm}.

Given a word $w$ in $\Delta$ and $S$ satisfying
\ref{2d rule}, \ref{end rule} and \ref{r rule},
we have the process $\pi(w)$ (see Definition~\ref{pi(w) and E(t,k,w)}).
We change this slightly to the process $\pi_G(w)$, which is the same
as $\pi(w)$, except that, when deletion of $j$ from $E_k$ is required in
$\pi(w)$, then no action is taken unless $j\in G_k$.
\begin{lemma}\label{G rules}
\ref{barrier constant rule},
\ref{block rule} and
\ref{barrier value inequality rule}
can be checked within the process $\pi_G(w)$.
\end{lemma}
\begin{proof}
We fix $k$, with $1\le k\le 2d$.
By the time the process $\pi(w)$ has completed,
$\alpha$ elements have been
deleted from $G_k$, and $\beta$ elements from $\setN \setminus G_k$, where
$\alpha \le \alpha+\beta \le r$, so
at most $r$ elements have been deleted.
Since $|G_k|=r+1$, one or more elements of
$G_k$ is never deleted during $\pi(w)$.
This means that we can compute $\mj(t,k,w)$ (defined in
\ref{barrier constant rule}) within $\pi_G(w)$.
So we can also compute $\bv(t,k,w)$.
We are therefore able to check \ref{barrier constant rule} and
\ref{barrier value inequality rule}.

We assume now that $w$
fails \ref{block rule} but
passes all the other rules.
We show that the failure can
be seen within the process $\pi_G(w)$.
Using the notation in which \ref{block rule} was stated, we
suppose the rule fails at time $t$, when the focus is at $k$.
We are assuming that $w=w_t.v_t$, where $w_t$ has length $t$,
and that $v_t$ deletes the first element of $B_k(t,w)$, but not all of
it. This means that $v_t$ starts with $\Delta^\alpha.S$, where
$1\le\alpha<b(t,w)$.
The number of discards from $G_k$ caused by $w_t$ is bounded above by
$\delta(t,w)$, and $\delta(t,w)+\alpha \le r$, by \ref{r rule}.
Therefore $w=w_t.v_t$ deletes no more than $r$ elements from $G_k$, so
that at least one element remains in $G_k$.
Since $b(t,w)>\alpha$, the smallest such element must be in the block
$B_k(t,w)$, so that we observe the failure of \ref{block rule} within
$G_k$.
\end{proof}

\section{Data structure}
\label{data structure}
We start with the given input: the $N\times 2d$-matrix $\sR$ as in \ref{data}.
Given $r$ with $0< r < N$, we want to generate all words in $\sW$.
In particular, this means that we need the ability to efficiently
check the rules, avoiding unnecessary computation.
We now describe the different parts of the data structure. As stated above,
significant parts of the data structure, used and changed by $F(p,s)$, are
restored as $F(p,s)$ exits. We will label these with the label
\texttt{restore}.

\begin{enumerate}[label=\thesection.\arabic{*}]
\item\label{tree}
{\tt Tree:}
The tree $\script T(r)$ of all words satisfying our first three rules
was introduced in Section~\ref{brief sketch}.
The subtree $\sT \subseteq \script T(r)$ is formed from \sW, the set
of words satisfying all six of our rules.
{\tt Tree} is a subtree of $\script T(r)$ that normally grows inside
\sT as our algorithm progresses, and various words are explored.
From time to time we may have to explore words outside \sW, because it
may not be immediately apparent that \ref{barrier value inequality
rule} must inevitably fail, no matter how the word is extended.
The error is eventually discovered, and
the illegal branch of {\tt Tree} is snipped off.

Each node $\nu$ of the tree corresponds to a unique word $w(\nu)$ in
the terms $\Delta$ and $S$, read from left to right as we descend the tree.
Each node of the tree is provided with data under the following headings:
\item \texttt{Nodes} and pointers $\Delta$, \texttt{ S, parent}:
We establish two symbolic values \tU and \tI.
The node $\nu$ has pointers
$\nu.\Delta$, {\tt $\nu$.S} and {\tt $\nu$.parent},
each pointing to another node of the tree or set equal to \tU or to \tI.
As soon as we find out that $w(\nu)$ is not a proper prefix of a word in $\sW$,
we set $\nu.\Delta=\tI$, and $\nu.\tS=\tI$.

Initially, {\tt Tree} has only one node, namely {\tt root},
representing the null word, and initially
all its three pointers point to \tU.
For every node $\nu$ the {\tt parent} pointer is never changed,
so that the {\tt parent} pointer of {\tt root} always points to \tU.
If $\nu\neq \mathtt{root}$, then its {\tt parent} pointer always
points to a node, not to \tU, so this can be used to recognize {\tt root}.
\item
{\tt CurrentNode (restore)}: is a global variable, pointing to the node
$\nu$ of {\tt Tree} currently reached by the algorithm.
It fixes the word $w=w(\nu)$ corresponding to this node.
Initially, {\tt CurrentNode} points to {\tt root}, the only
initial node of {\tt Tree}.
\item \texttt{BSZ (restore)} and $Q$ (constant):
Let $B_i$ be the set of $G_i$-blocks, arranged in order using
$\prec_i$. We set $b_i=|B_i|$, and number the elements of $B_i$
in order, using {\it block numbers}
$\set{1,\dots,b_i}$, with 1 numbering the $\prec_i$-smallest
block and $b_i$ the $\prec_i$-largest block.
$Q$ is an $N\times 2d$ matrix,
whose entry in position $(j,i)$ gives the block number of the
$G_i$-block containing the rectangle index $j$, provided $j\in
G_i$, and is equal to \tU if $j\notin G_i$.
Constructing $Q$ takes $O(d\,N)$ steps and
$Q$ does not change during the algorithm.
The changing block sizes are recorded in the $(r+1)\times
2d$-matrix \tB, such that $\tB(k,i)$ is the block size of the $k$-th
block of $G_i$ if $k\le b_i$. If $k>b_i$, then $\tB(k,i)=\tU$.
Initializing \tB takes $O(d\,r)$ steps.
Updating \tB when the rectangle $R_j$ is discarded or restored takes $O(d)$ steps.
\item \tD (\texttt{restore} Doubly Linked List):
In Section~\ref{selection versus sorting} we defined,
for each $i$ with $1\le i\le 2d$,
$G_i$ as the $\preccurlyeq_i$-ordered set of the $r+1$
rectangle indices that are smallest with respect to $\preccurlyeq_i$.
\tD is an $N\times 2d$-matrix, whose $k$-th column is denoted by $\tD_k$.
Each entry of \tD is a pair of rectangle
indices, that initially records $G_i$ and its order as follows.
If $j\notin G_i$, the entry at row $j$ and column $i$ is \tU
and this entry is unchanged throughout the algorithm.
If $j\in G_i$, the $(j,i)$ entry is initially a pair (\tp,\tn)
of rectangle indices,
specifying the ordering $\preccurlyeq_i$ on $G_i$ in the usual manner.
We adjoin a row 0 to \tD to accommodate, for each $i$, a start
entry at $(0,i)$ with
$\tp=\tU$ and \tn equal to the $\preccurlyeq_i$-smallest surviving
element of $G_i$.
If $j$
is the $\preccurlyeq_i$-largest rectangle index surviving, then
$\tn=\tU$ at position $(j,i)$.
The entries in \tD are changed in the usual way (taking
constant time) as elements $j\in G_i$ are
deleted.
When $j$ is deleted
from $G_i$, we do not change the temporarily meaningless
values (\tp,\tn) at position $(j,i)$; instead these values are
retained for use when $j$ is restored.
\tD takes $O(d\,(N+r\log r))$ steps to initialize. Deleting or
restoring a rectangle index $j$ takes $O(d)$ steps.
\item \label{nu.bv} $\nu.\mathtt{vol}$:
Let $m$ be the number of terms in the word $w=w(\nu)$ that are equal to $S$.
For $k\le m$, $\bv(k,w)=r_{j,k}$,
where $j$ is the $\preccurlyeq_k$-smallest
element remaining in the doubly linked list $\tD_k$---this is because
\ref{barrier constant rule} is enforced at each step.
For $2i \le m$, let
$$ \alpha(i)= \bv(2i,w)-\bv(2i-1,w).$$
This is already known, unless $m=2i$ and the
final term of the word $w(\nu)$ is $S$.
If $\alpha(i)<0$, we set $\nu.\mathtt{parent}.\tS = \tI$ and snip off
the part of the tree hanging from $\nu$.
Otherwise set
$$\nu.\mathtt{vol} = \prod \set{\alpha(i)|\ 2i\le m \text{ and }
\alpha(i)>0},
$$
which can be evaluated inductively, each step of the induction
taking constant time.
\item $\nu.\dim$: We set $\nu.\mathtt{dim}$ equal to the number
of $i$ such that $\alpha(i)>0$.
\item {\tt Results:}\label{Results}
For $0\le p \le r$, $\mathtt{Results}(p)$ has three
subcomponents, namely $\mathtt{dim}$, $\mathtt{vol}$ and
$\mathtt{address}$.
The first two subcomponents give the lexicographically
largest result of the form {\tt(dim,vol)} so far achieved
by discarding $p$ $d$-rectangles from the list $\sR$, then taking
the intersection of the remaining $d$-rectangles.
The third subcomponent is the address of the
{\tt Tree} node $\nu_p$ that was current at
the time that this best result was found.
Initially each {\tt dim} entry is $-1$, each {\tt vol} entry is
\tU, and each {\tt address} entry is equal to \tU.
Initialization takes $O(r)$ steps.
\end{enumerate}

\section{Pseudocode}
\label{pseudocode}
{%

Here is pseudocode for the algorithm sketched in Section~\ref{brief
sketch}.
We will use \& and * for the address and indirection operators, as in
C. The address can be, as in C, an address in memory,
or, in languages that do not attempt to mimic machine architecture,
the number of a row in a table or matrix.
The code starts with initialization, which consists of the
sorting processes described above in Section~\ref{data structure}
and setting up the various data structures in the obvious way. This
takes $O(dN + dr\log r)$ steps.

\vspace*{10pt}

The execution path of the program $F(p,s)$, whose description
follows, depends on the {\it state} in which it starts. By the {\it
state} we mean the values contained in the global data structures described
above. So $F(p,s)$ stands for many different possible execution
paths. However, We will be able to give a reasonable upper bound
for the number of steps that $F(p,s)$ requires, and this
upper bound is independent of the particular execution path.

\ref{2d rule} is enforced by the inequality $0\le s \le 2d$.
We need to check that $F(p,s)$ enforces the other five rules.

\subsection{$F(p,s)$: deciding how the tree can be extended}
\label{deciding}
\alginout{The global variables described above;
$p$, an upper bound for the number of characters $\Delta$
still to be generated, or, equivalently, the number of rectangles still to
discard; $s$, the number of characters $S$ still to be generated, or,
equivalently, $2d$ minus the number of characters $S$ already generated; {\tt
CurrentNode}, the address of the current node $\nu$ in the tree.}
{Assignment of the pointers $\nu.\Delta$ and $\nu.\tS$ to
\tI, where appropriate. Storage of best results so far.}

\vspace*{10pt}
\begin{algtab}
\algbegin
$f\leftarrow 2d-s+1$;\algcomm{focus now at $f$}\alglabel{f}\\
$j\leftarrow$ smallest element of $\tD_f$;\alglabel{j}\\
$\nu \leftarrow *\mathtt{CurrentNode}$;\algcomm{$\nu$ is the current node}\\
\algif{$s==0$} \algcomm{Exploration of this word is complete and
all rules are satisfied.}
	set $\nu.\Delta=\nu.\tS=\tI$;\algcomm{Enforcing \ref{end
	rule}}\alglabel{s nonzero}\\
	\algif{$\mathtt{Results}(r-p).\mathtt{(dim,vol)}<
		\nu.\mathtt{(dim,vol)}$ (lexicographically)}
		$\mathtt{Results}(r-p).(\mathtt{(dim,vol)}\leftarrow
		(\nu.\mathtt{(dim,vol)},\&\nu$);\\
	\algendif
 \algelse
 \algcomm{
	 $s>0$. Check all rules, but without altering the main data structures.}
	\algif{$p==0$}
		$\nu.\Delta\leftarrow\tI$;\algcomm{Enforcing \ref{r
		rule}}\alglabel{partial nonzero}\\
	\algelsecomm{$p >0$}
		\algif{$\tB(Q(j,f),f) > p$}
			\algcomm{Note that $\tB(Q(j,f),f)$ is the
			block containing $j$. When $p=0$,
			the inequality is automatically
			satisfied, we have already discarded $r$ rectangle
			indices, and further discards are not allowed.
			If $p>0$ and the inequality is satisfied then
			\ref{block rule} would be violated by a further
			discard and this rules out continuations of $w(\nu)$
			with next letter $\Delta$.}\\
			$\nu.\Delta\leftarrow\tI$\\
		\algelsecomm{\ref{block rule} is satisfied and we check \ref{barrier constant rule}}
			\algforeach{$i$ such that $1\le i < f$}
				$k\leftarrow$ smallest element of $\tD_i$;\\
				\algif{$Q(j,i)==Q(k,i)$ \algand $\tB(Q(j,i),i)==1$}
					\algcomm{$j,k$ in same $G_i$-block and
					\ref{barrier constant rule} would be
					violated by a further deletion.}\\
					$\nu.\Delta\leftarrow\tI$;\\
					\algbreak\ from \textbf{foreach} loop\\
				\algendif
			\algendfor
		\algendif
	\algendif
	\algif{$f$ is even}
		$k\leftarrow$ smallest element of $\tD_{f-1}$\\
		\algif{$r_{k,f-1} > r_{j,f}$}
			\algcomm{\ref{barrier value inequality rule} would be
			violated by following with \tS.}\\
			$\nu.\tS\leftarrow\tI$;\alglabel{negative volume}\\
		\algendif
	\algendif
\algendif
\algend
\end{algtab}

\subsection{$F(p,s)$: extending the tree, and recursion}
\label{recursion}
At this stage we have determined all
failures in the rules that would be immediately apparent on
following $w(\nu)$ with either $\Delta$ or \tS. As a result, we know
which new nodes to construct, and we proceed with this task.
We also alter the global data structures as we explore the tree.

\alginout{The global variables;
$p$, an upper bound for the number of characters $\Delta$
still to be generated, or, equivalently, the number of rectangles still to
discard; $s$, the number of characters $S$ still to be generated, or,
equivalently,
$2d$ minus the number of characters $S$ already generated; {\tt
CurrentNode}, the address of the current node $\nu$ in the tree.}{Construction of new nodes of the tree. Determination of
the best results possible so far, using recursion.}

\vspace*{10pt}
\begin{algtab}
\algbegin
$f\leftarrow 2d-s+1$;
\algcomm{same as Line~\algref{f} of \ref{deciding}}\\
$j\leftarrow$ smallest element of $\tD_f$;
\algcomm{same as Line~\algref{j} of \ref{deciding}}\\
$\nu\leftarrow *\mathtt{CurrentNode}$\\
\algif{$\nu.\Delta\neq \tI$}
	Create node $\alpha$;\\
	$\mathtt{Treelocation} \leftarrow \&\alpha$;\\
	$\nu.\Delta\leftarrow\&\alpha$;\\
	$\alpha.\mathtt{parent}\leftarrow\&\nu$;\\
	\algforeach{$i$ such that $j\in \tD_i$}\alglabel{tD}
		Delete $j$ from $\tD_i$;\\
		Decrease $\tB(Q(j,i),i)$ by 1;\alglabel{decrease block}\\
	\algendfor
	$F(p-1,s)$;\algcomm{$p>0$ by
	Line~\algref{partial nonzero} of \ref{deciding}}\\
	\algif{$\alpha.\Delta==\tI$ \algand $\alpha.\tS==\tI$}
		set $\nu.\Delta = \tI$;\\
	\algendif
	Reinsert $j$ in the doubly linked lists $\tD_i$;\\
	Increase by 1 the various $\tB(Q(j,i),i)$
	(see Line~§\algref{decrease block} above);\\
	$\mathtt{CurrentNode}\leftarrow \&\nu$;\\
\algendif
\algif{$\nu.\tS\neq \tI$}
	Create node $\beta$;\\
	$\mathtt{Treelocation} \leftarrow \&\beta$;\\
	$\nu.\tS\leftarrow\&\beta$;\\
	$\beta.\mathtt{parent}\leftarrow\&\nu$;\\
	\algforeach{$i$ such that $j\in \tD_i$}
		Delete $j$ from $\tD_i$;\\
		Decrease $\tB(Q(j,i),i)$ by 1;\\
	\algendfor
	\algif{$f$ is even}
		$k\leftarrow$ smallest element of $\tD_{f-1}$;\\
		\algif{$r_{k,f-1} == r_{j,f}$}
			$\beta.(\mathtt{dim,vol})=\nu.(\mathtt{dim,vol})$;\\
		\algelsecomm{By Line~\algref{negative volume} of
		Subsection~\ref{deciding} $r_{k,f-1} < r_{j,f}$.}
			$\beta.\mathtt{vol} \leftarrow \nu.\mathtt{vol}\times
			\bracket{r_{j,f}-r_{k,f-1}}$;\\
			$\beta.\mathtt{dim} \leftarrow \nu.\mathtt{dim}+1$;\\
		\algendif
	\algendif
	$F(p,s-1)$;\algcomm{$s>0$ by Line~\algref{s nonzero} of
	Subsection~\ref{deciding}}\\
	\algif{$\beta.\Delta==\tI$ \algand $\beta.\tS==\tI$}
		set $\nu.\tS = \tI$;\\
	\algendif
	Reinsert $j$ in the doubly linked lists $\tD_i$;\\
	Increase by 1 the various $\tB(Q(j,i),i)$
	(see Line~§\algref{decrease block} above);\\
	$\mathtt{CurrentNode}\leftarrow \&\nu$;\\
\algendif
\algend
\end{algtab}
}
\section{Estimate for number of steps}

By induction on $p+s$, there is a unique function
$f(p,s)$ with the following properties:
\begin{equation}\label{f equations}
\begin{split}
f(p,0) &= 1 \text{ for } p \ge 0\,;\\
f(0,s) &= s+1 \text{ for } s\ge 0\,;\\
f(p,s)&= f(p-1,s)+f(p,s-1)+d
\text{ for } p > 0 \text{ and } s>0 \,.
\end{split}
\end{equation}
From the pseudocode in Section~\ref{pseudocode}, one sees
that there is a constant $k>0$, independent of $N$, $d$ and $r$,
 such that the number of steps necessary
to execute $F(p,s)$ is bounded by $k.f(p,s)$.

Using the addition formula for binomial coefficients
$$\binom{n}{s}=\binom{n-1}{s}+\binom{n-1}{s-1} ,$$
one checks that $f(p,s)$ is given by
\begin{equation}\label{steps for F}
f(p,s)=d\binom{p+s}{s} + \binom{p+s+1}{s}-d ,
\end{equation}
since it satisfies Equation~\ref{f equations}.
In particular, we are interested in
\begin{equation}\label{f r d bound}
f(r,2d) = d\binom{r+2d}{2d} + \binom{r+2d+1}{2d}-d.
\end{equation}

To complete the computation of the upper bound for the number of steps
involved in our algorithm, we need to extract the information in the
data structure {\tt Results} (see \ref{Results}).
From $\rho=\mathtt{Results}(p)$, we find the set of discarded
$d$-rectangles by following {\tt parent} pointers, starting at
$\rho.\mathtt{address}$.
This takes at most $O(p+2d)$ steps.
Summing over $0\le p \le r$, we obtain a bound
\begin{equation}\label{reading Results}
O((r+1)2d + r(r+1)/2).
\end{equation}
steps, where the constant implicit in the $O$
notation is independent of $N$, $d$ and $r$.
Combining this with the bounds of Equations~\ref{selection} and \ref{f
r d bound}, we get an overall bound
\begin{equation}\label{overall}
O\bracket{d\binom{r+2d}{2d} + dN}
\end{equation}
since the omitted terms are dominated by those that remain in
Equation~\ref{overall}.

\section{The dual situation}
Let $\script S$ be the set of closed
non-empty $d$-rectangles in $\reals^d$, together with the null-set.
We have a lattice: the greatest lower bound of a finite subset
of $\script S$ is given by intersection. The least upper bound of any
finite subset of $\script S$ exists, but is not in general equal to the union.

Theorem~\ref{main} applies equally well to the dual situation. Here we
are given $N$ $d$-rectangles and a number $r$ with $0\le r\le N$. The
task is to find which $r$ rectangles to discard, so that the
smallest rectangle containing
all except these $r$ rectangles is as small as possible.

We can then prove the dual result to Theorem~\ref{main} by reversing
all inequalities and replacing intersection by least upper bound. It
is possible to sketch scenarios in which such a result might be used,
but we spare the reader rather than labour the point.

\bibliographystyle{plain}
\bibliography{all}

\begin{thebibliography}{1}

\bibitem{Blum.etal:1973:Time}
M.~Blum, R.W. Floyd, V.~Pratt, R.L. Rivest, and R.E. Tarjan.
\newblock Time bounds for selection.
\newblock {\em Journal of Computer and System Sciences}, 7(4):448--461, 1973.

\bibitem{Hoare:1961:Algorithm}
C.A.R. Hoare.
\newblock Algorithm 63: partition algorithm 65: {F}ind.
\newblock {\em Communications of the ACM}, 4(7):321--322, 1961.

\bibitem{Raza.etal:2012:RAMTaB:}
Shan-e-Ahmed~E. Raza, Ahmad Humayun, Sylvie Abouna, Tim~W. Nattkemper,
  David~{B}.{A}. Epstein, Michael Khan, and Nasir~M. Rajpoot.
\newblock Ramtab: Robust alignment of multi-tag bioimages.
\newblock {\em Plos One}, 7(2):e30894--e30894, 2012.

\bibitem{Schubert.etal:2006:Analyzing}
Walter Schubert, Bernd Bonnekoh, Ansgar~J. Pommer, Lars Philipsen, Raik
  Böckelmann, Yanina Malykh, Harald Gollnick, Manuela Friedenberger, Marcus
  Bode, and Andreas W.~M. Dress.
\newblock Analyzing proteome topology and function by automated
  multidimensional fluorescence microscopy.
\newblock {\em Nature biotechnology}, 24(10):1270--8, 2006.

\end{thebibliography}
\end{document}